\documentclass[a4paper,UKenglish,cleveref, autoref, thm-restate]{lipics-v2021}

\hideLIPIcs  

\DeclareMathOperator{\dist}{d}
\newcommand{\R}{\mathbb{R}}

\graphicspath{{./figures/}}

\title{Pattern Formation for Fat Robots with Lights} 

\author{Rusul J. Alsaedi}{University of Sydney, Australia}{rals2984@uni.sydney.edu.au}{}{}

\author{Joachim Gudmundsson}{University of Sydney, Australia}{joachim.gudmundsson@sydney.edu.au}{https://orcid.org/0000-0002-6778-7990}{}

\author{André {van Renssen}}{University of Sydney, Australia}{andre.vanrenssen@sydney.edu.au}{https://orcid.org/0000-0002-9294-9947}{}

\authorrunning{R.\,J. Alsaedi, J. Gudmundsson, and A. van Renssen} 

\Copyright{Rusul J. Alsaedi, Joachim Gudmundsson, and André van Renssen} 

\ccsdesc[500]{Theory of computation~Computational geometry} 

\keywords{Pattern formation, Robots with lights, Fat robots, Obstructed visibility, Collision avoidance} 

\nolinenumbers 

\EventEditors{John Q. Open and Joan R. Access}
\EventNoEds{2}
\EventLongTitle{42nd Conference on Very Important Topics (CVIT 2016)}
\EventShortTitle{CVIT 2016}
\EventAcronym{CVIT}
\EventYear{2016}
\EventDate{December 24--27, 2016}
\EventLocation{Little Whinging, United Kingdom}
\EventLogo{}
\SeriesVolume{42}
\ArticleNo{23}

\begin{document}

\maketitle

\begin{abstract}
Given a set of $n\geq 1$ unit disk robots in the Euclidean plane, we consider the Pattern Formation problem, i.e., the robots must reposition themselves to form a given target pattern. This problem arises under obstructed visibility, where a robot cannot see another robot if there is a third robot on the straight line segment between the two robots. 
Recently, this problem was solved in the asynchonous model for fat robots that agree on at least one axis in the robots with lights model where each robot is equipped with an externally visible persistent light that can assume colors from a fixed set of colors~\cite{bose2020arbitrary}. In this work, we reduce the number of colors needed and remove the axis-agreement requirement in the fully synchronous model. In particular, we present an algorithm requiring 7 colors when scaling the target pattern is allowed and an 8-color algorithm if scaling is not allowed. Our algorithms run in $O(n) + O(q \log n)$ rounds with probability at least $1 - n^{-q}$.
\end{abstract}

\section{Introduction}
We consider a set of $n\geq 1$ unit disk robots in the two-dimensional Euclidean plane and aim to position these robots to form a given pattern. This problem is commonly referred to as the Pattern Formation problem. The robots work under the classical oblivious robots model \cite{Flocchini2012} which means they are autonomous (no external control), anonymous (no unique identifiers), indistinguishable (no external markers), history-oblivious (no memory of activities done in the past), silent (no means of direct communication), and possibly disoriented (no agreement on their coordinate systems). All robots execute the same algorithm and operate following Look-Compute-Move (LCM) cycles \cite{DAS2016171} (i.e., when a robot becomes active, it uses its vision to get a snapshot of its surroundings (Look), computes a destination point based on this snapshot (Compute), and finally moves towards the computed destination (Move)). In this paper, the robots are scheduled using the fully synchronous model, where there is a (discretized) notion of \emph{rounds} and in every round all robots are activated. 

Throughout the long history of this classical robot model many applications have been considered. Examples include coverage, exploration, intruder detection, data delivery, and symmetry breaking problems~\cite{cieliebak2012distributed}. Unfortunately, frequently the robots in this model were considered to be point robots which do not occupy any space.

This assumption has an important consequence in terms of visibility. When working with point robots, unobstructed visibility is a common assumption. This means that any three collinear robots are mutually visible to each other. However, in reality robots are not dimensionless points and hence they may block the view of other (collinear) robots. Recently, the capabilities of these robots under obstructed visibility has been studied  \cite{Agathangelou2013,https://doi.org/10.48550/arxiv.2206.14423,Bolla2012,Chaudhuri15,Cohen2008,Cord-Landwehr2011,Czyzowicz2009,di2017mutual,Dutta2012,Flocchini2015,Luna2014,Luna2014b,Sharma2015b,Sharma2015,Sharma2015c,Vaidyanathan2015}. Under obstructed visibility, robot $r_i$ can see robot $r_j$ if and only if there is no third robot on the straight line segment connecting $r_i$ and $r_j$.

One of the problems studied under obstructed visibility (and the focus of this paper) is the Pattern Formation problem: Starting from arbitrary distinct positions in the plane, determine a schedule to reposition the robots such that they form the given (target) pattern without collisions \cite{bose2020arbitrary,bose2021arbitrary,vaidyanathan2022fast}.
We say that two robots collide if at any time they share the same position. In previous work, the target pattern is allowed to be scaled, rotated, translated, and reflected. 

To tackle this and other robot problems, a generalization of the classical model has recently become the focus of significant interest~\cite{https://doi.org/10.48550/arxiv.2206.14423,bose2020arbitrary,bose2021arbitrary,di2017mutual,Luna2014,Luna2014b,Peleg2005,Sharma2015b,Sharma2015,Vaidyanathan2015,vaidyanathan2022fast}. This variant, called the luminous robots model (or robots with lights model), equips robots with an externally visible light which can assume different colors from a fixed set. These lights are persistent, i.e., the color of the light is not erased at the end of the LCM cycle. Except the assumption of the availability of lights, the robots work similarly to the classical model. This model corresponds to the classical oblivious robots model when the number of colors in the set is 1 \cite{di2017mutual,Flocchini2012}. One objective in this model is to solve the problem while minimizing the size of the color set.

\subsection{Related Work}
There has been considerable work on the Pattern Formation problem for point robots \cite{bose2022distributed,bose2020pattern,chaudhuri2014pattern,fujinaga2015pattern,ghosh2022move,kundu2022arbitrary,pattanayak2020distributed} and some of these also considered the lights model while solving the problem \cite{bose2021arbitrary,vaidyanathan2022fast}. Other work in this area includes that by Cicerone {\it et al.}~\cite{cicerone2021solving}, who presented an algorithm to solve the Pattern Formation problem for point robots with chirality (the robots agree on the orientation of the axes, i.e., on the meaning of clockwise), and Flocchini {\it et al.}~\cite{flocchini2008arbitrary}, who studied the problem for point robots in the asynchronous setting, but they required that the robots agree on their environment and observe the positions of the other robots. While these works provided techniques to overcome various difficulties, they did not take the physical extents of the robots into account. Unfortunately, the techniques developed for point robots cannot be applied directly to solve the Pattern Formation for fat robots, due to the effect these extends have on collision avoidance. 

The work most closely related to our result is due to Bose {\it et al.}~\cite{bose2020arbitrary}. They studied the Pattern Formation problem for fat robots in the robots with lights model and used 10 colors in the asynchronous model. Unfortunately, their solution assumes that all robots agree on an axis of the coordinate system. Our solutions remove this assumption and reduces the number of colors required in the fully synchronous model. 

Other related work for fat robots includes the work by Kundu {\it et al.}~\cite{kundu2022arbitrary}, who studied the Pattern Formation problem for fat robots with lights on an infinite grid with one axis agreement. They solved the problem using 9 colors. Unfortunately, they did not bound the running time of their algorithm. 

\subsection{Contributions}
We first present two algorithms using at most 11 colors, which through careful analysis we improve to use only 7 colors when scaling of the pattern is allowed and 8 if this is not allowed. None of our algorithms require the pattern to be rotated, translated, or reflected. Our algorithms require $O(n) + O(q \log n)$ rounds. Here $q>0$ is related to leader election, which can be solved in $O(q \log n)$ rounds with probability at least $1-n^{-q}$~\cite{vaidyanathan2022fast}. 

Our algorithms work under the fully synchronous model and are collision-free. Interestingly, unlike previous work, our algorithms do not require any additional assumptions on the capabilities of the robots or any shared information or coordinate system. The moves of the robots are rigid, i.e., an adversary does not have the power to stop a moving robot before reaching its computed destination~\cite{Flocchini2012}.

One of the main challenges to obtain these results is that the locations of the robots in the provided input pattern can be very close together, leading to the potential issue of a leader robot being unable to leave the pattern constructed so far, when we try to use the well-known method of having the leader move the robots in place by having them follow it. We address this issue in two ways. First, if we can scale the pattern with respect to the coordinate system of the leader, we do this in such a way that we guarantee that there is always enough space for the leader to lead a robot to its intended location, and leaving sufficient additional space for the leader to leave the pattern again afterwards. If this form of scaling is not allowed, we developed a novel method of the leader ``pushing'' robots instead of having them follow the leader. This allows us to move the robots in place without the leader moving into the pattern itself, thus avoiding the issue of it having insufficient space to move out again altogether. 

\section{Preliminaries}
Consider a set of $n\geq 1$ anonymous robots $R=\{r_1,r_2,\ldots,r_n\}$ operating in the Euclidean plane $\R^2$. The number of robots $n$ is not assumed to be known to the robots.
We assume that each robot $r_i\in R$ is a non-transparent disk with diameter 1. The center of the robot $r_i$ is denoted $c_i$, and the position of $c_i$ is also said to be the position of $r_i$. We denote by $\dist(r_i,r_j)$ the Euclidean distance between the two robots $c_i$ to $c_j$. For simplicity, we use $r_i$ to denote both the robot $r_i$ and the position of its center $c_i$. 
Each robot $r_i$ has its own coordinate system, and it knows its position with respect to this coordinate system.
Robots may not agree on the orientation of their coordinate systems. However, since all the robots are of unit size, they implicitly agree on the notion of unit length. 

Each robot has a camera to take a snapshot of the plane and the distance that is visible to this camera is infinite, provided that there are no obstacles blocking its view (i.e., another robot) \cite{Agathangelou2013}. 
Following the fat robot model \cite{Agathangelou2013,Czyzowicz2009},
we assume that a robot $r_i$ can see another robot $r_j$ if there is at least one point on the bounding circle of $r_j$ that is visible to $r_i$.
Similarly, we say that a point $p$ in the plane is visible to a robot $r_i$ if there is a point $p_i$ on the bounding circle of $r_i$ such that the straight line segment $\overline{p_i p}$ does not intersect any other robot. 

Each robot $r_i$ is equipped with an externally visible light that can assume any color from a fixed set $C$ of colors. The set $C$ is the same for all robots. 
The color of the light of robot $r$ at time $t$ can be seen by all robots that are visible to $r$ at time $t$.

At any time $t$, a robot $r_i\in R$ is either active or inactive. When active, $r_i$ performs a sequence of {\em Look-Compute-Move} (LCM) operations:
\begin{itemize}
\item {\em Look:} the robot takes a snapshot of the positions of the robots visible to it in its own coordinate system; 
\item {\em Compute:} executes its algorithm using the snapshot which returns a destination point $x\in \R^2$ and a color $c\in C$; and
\item {\em Move:} moves to the computed destination $x\in \R^2$ (if $x$ is different than its current position) and sets its own light to color $c$.
\end{itemize}

Each robot executes the same algorithm locally every time it is activated and a robot's movement cannot be interrupted by an adversary. Two robots $r_i$ and $r_j$ are said to {\em collide} at time $t$ if the bounding circles of $r_i$ and $r_j$ share a common point at time $t$. To avoid collisions among robots, we thus have to ensure that at all times $\dist(r_i,r_j)\geq 1$ for any robots $r_i$ and $r_j$.

We assume that the execution starts at time $0$. At this time, the robots start in arbitrary positions with $\dist(r_i,r_j)\geq 1$ for any two robots $r_i$ and $r_j$, and the color of the light of each robot is set to {\it off}.

The Pattern Formation problem is now defined as follows: Given any initial positions of the robots, the robots must reposition themselves to form a given target pattern without having any collisions in the process. The target pattern is given as input as a list of locations, one for each robot in the pattern and allowed to be scaled, rotated, translated, and reflected. An algorithm is said to solve the Pattern Formation problem if it always achieves the target pattern from any initial configuration. The pattern should be constructed if this is indeed possible. We measure the quality of the algorithm using the number of distinct colors in the set $C$ and the number of rounds needed to solve the Pattern Formation problem.

\section{Algorithms}
In this section, we present algorithms that solve the Pattern Formation problem for $n \geq 1$ robots of unit disk size in the robots with lights model. Our algorithms assume the fully synchronous setting. Our algorithms execute four phases: Mutual Visibility, Leader Election, Line Formation, and Pattern Formation. We first present two algorithms that use at most 11 colors: one algorithm allows the target pattern to be scaled, while the other does not. 
Neither algorithm requires the target pattern to be rotated or reflected. We note that these operations are considered with respect to the coordinate system of the robot elected leader in the Leader Election phase. In Section~\ref{sec:improving_colors} we revisit the number of colors needed by our algorithms to obtain our final results. 

\subsection{Mutual Visibility}
Starting from any initial configuration, the goal of this phase is to move the robots to positions where every robot can see all other $n-1$ robots. Existing work \cite{https://doi.org/10.48550/arxiv.2206.14423} achieves this by positioning all robots on a convex hull and the presented algorithm runs in $O(n)$ rounds and uses only two colors (\emph{off} for robots that are not yet a corner of the convex hull, and \emph{corner} for robots that are). This algorithm runs in the fully synchronous setting with obstructed visibility in the robots with lights model and avoids collisions. While we use this algorithm as a black box, the one important thing to note is that to ensure that there is always enough space for all robots on the boundary of the convex hull, the corner robots expand the convex hull in each round. Hence, until this phase is completed, the corner robots move each round. 
Figure {\ref{fig:3.1}} shows an initial configuration as well as the final situation of the Mutual Visibility phase.

\begin{figure}[ht]
 \centering
  \includegraphics{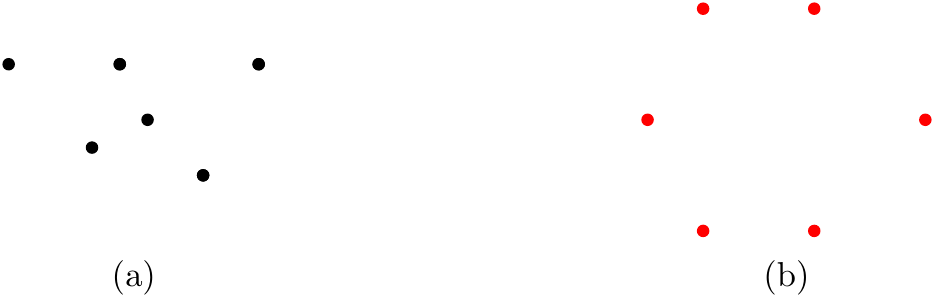}
  \caption{An example of the Mutual Visibility phase: (a) an initial configuration and (b) the end configuration. Throughout the paper the robots are shown as dots for simplicity.}
  \label{fig:3.1}
\end{figure}

\subsection{Leader Election}
The goal of the Leader Election phase is for a single robot to be elected as a leader. After the Mutual Visibility phase, every robot sees all other $n-1$ robots, so the robots know the value of $n$ (while they remain visible to each other, as they have no memory). It is known that electing a leader is not possible using a deterministic algorithm in an anonymous distributed system \cite{attiya2004distributed}. Hence, we use the randomized algorithm by Vaidyanathan {\it et al.}~\cite{vaidyanathan2022fast}. 

Initially, all robots are competing and use a color, say \emph{competing}, to indicate this. The algorithm proceeds in iterations until it finishes with a single leader. Each iteration has a constant number of rounds. In an iteration with $n$ competing robots, each robot flips a coin whose probability of success is $1 / n$. If a robot is successful, then the robot leaves its color as \emph{competing}. Otherwise, it changes its color to \emph{non-competing}. If there is exactly one competing robot left, the robots have successfully elected a leader and this robot changes its color to the \emph{leader} color. Otherwise, this iteration was unsuccessful and all robots change their color to back \emph{competing} and try again. 
Figure {\ref{fig:3.2}} shows an example of the Leader Election phase.

\begin{figure}[ht]
 \centering
  \includegraphics{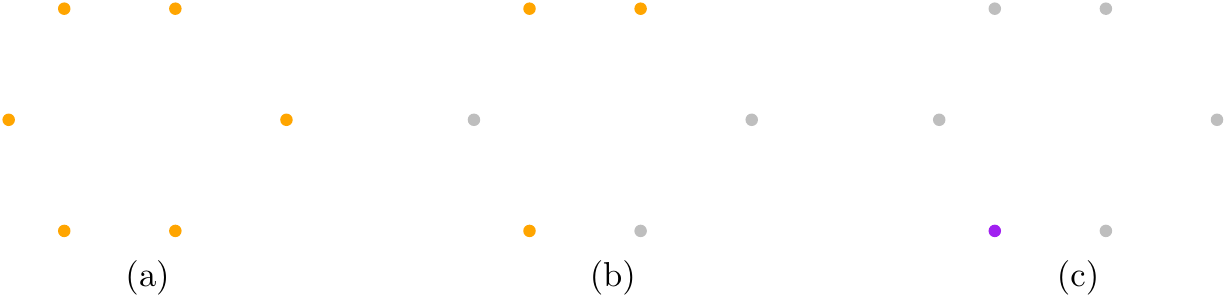}
  \caption{An example of the Leader Election phase: (a) initially all robots are competing (orange), (b) an unsuccessful iteration with competing and non-competing (gray) robots, and (c) a successful iteration, where a single robot is elected leader (purple).}
  \label{fig:3.2}
\end{figure}

\subsection{Line Formation} 
The goal of the third phase is to move the robots from their convex hull positions to a line away from both their old positions and away from where the leader will build the target pattern in the next phase. In this phase the leader will consider its own coordinate system and move the other robots one by one to achieve this goal. Where the leader forms the line of robots depends on the current location of the convex hull of robots and the area required to build the pattern (we assume without loss of generality that the leader will use its origin $(0,0)$ as the topright corner of the bounding box of the target pattern). 

Since the leader is a robot on the convex hull, at least one of the quadrants originating from the leader's position points away from the convex hull and thus does not contain any robots. We assume without loss of generality that this is the lowerleft quadrant in the leader's coordinate system. The leader will use this quadrant to build the line. Without loss of generality, we will explain the Line Formation and Pattern Formation phases using this quadrant. This assumption can easily be removed by mirroring the approach along the $x$- or $y$-axis.

The leader computes the topmost position of the line as follows. Let $x_{hull}$ and $y_{hull}$ denote the minimum $x$- and $y$-coordinate of any robot on the convex hull according to the leader's coordinate system. Similarly, let $x_{pattern}$ and $y_{pattern}$ denote the minimum $x$- and $y$-coordinate of the target pattern according to the leader's coordinate system. The leader now computes the topmost position $(x_{line}, y_{line})$ on the line as $x_{line}=\min (x_{hull},x_{pattern})-M$ and $y_{line}=\min (y_{hull},y_{pattern})-M$, where $M$ is some large constant. Picking this coordinate guarantees that the line will be built below and to the left of both the convex hull and the target pattern. We note that while the leader can compute this coordinate, it cannot store it for use in future rounds, so it needs to move directly from its current position to where it wants to build the line. 

If $(x_{line}, y_{line})$ has only a single closest robot on the convex hull, the leader now moves to this position $(x_{line}, y_{line})$ and changes its color to \emph{follow the leader}. As there are no robots on the line yet, this signals that the robot closest to the leader's position should move to this location. In the next round, the leader moves down one unit (and sets its color to \emph{do not follow}) to avoid colliding with the approaching robot. Once this robot reaches its position, it sets its color to \emph{on the line}. Figure {\ref{fig:3.3}} shows these movements as well as the remaining ones in this phase. If $(x_{line}, y_{line})$ has more than one closest robot on the convex hull, the leader moves to ensure that the leftmost bottommost robot is the unique closest one in the next round. It does this by computing an $x_{temp}$ such that this is the case and then setting $x'_{line}=\min (x_{temp}, x_{line})$ to get new coordinates $(x'_{line},y_{line})$ ensuring that the moves described earlier activate a single robot to move to the leader's position. Once the first robot is in place, this robot will not move for the remainder of the phase. 

\begin{figure}[ht!]
 \centering
  \includegraphics[scale=0.95]{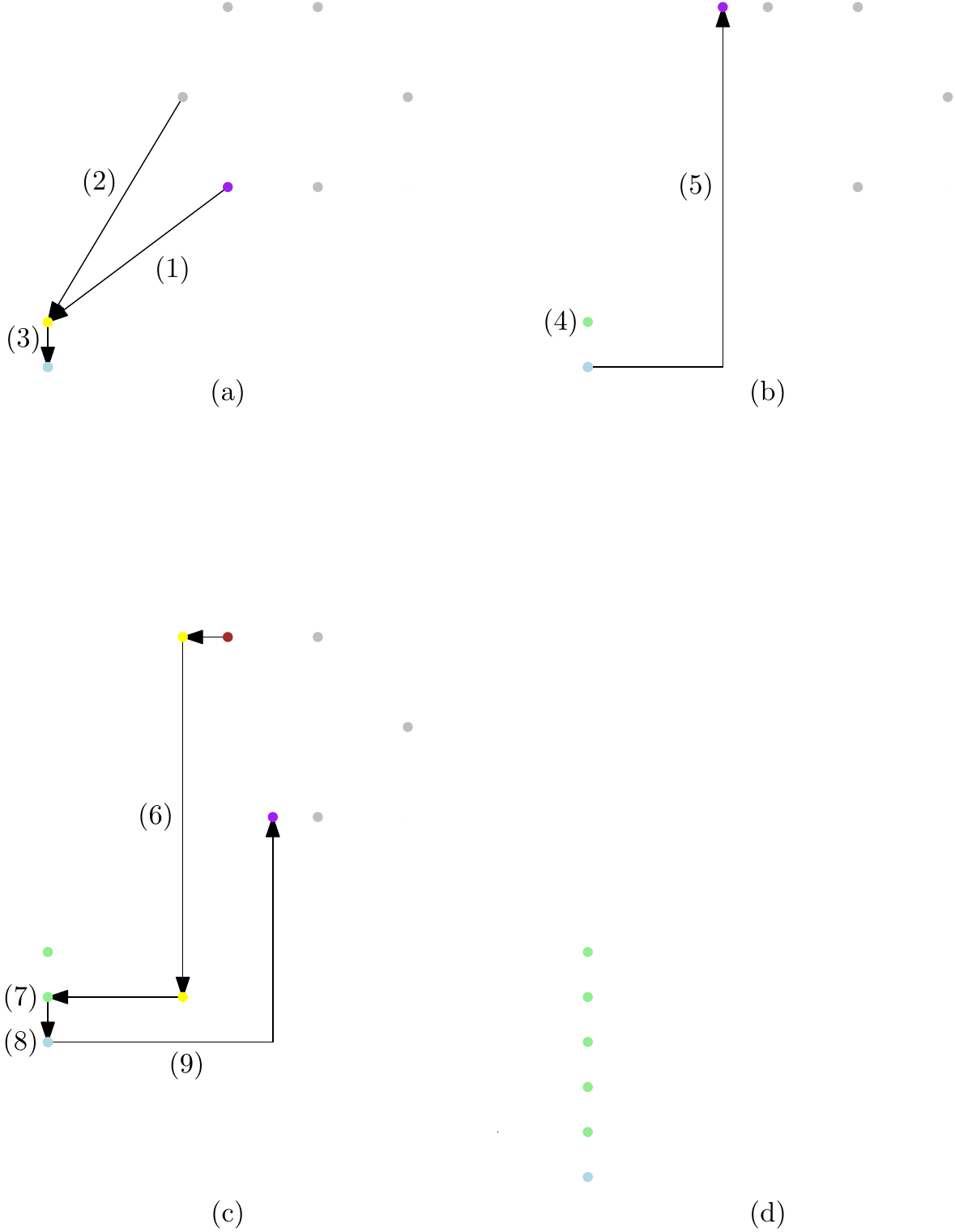}
  \caption{An example of the Line Formation phase (numbers indicate order of operations): (a) the leader (purple) moves to the first position on the line and signals (yellow) the closest robot to move there in the next round before moving out of the way and setting its color to \emph{do not follow} (lightblue), (b) the first robot changes its color to \emph{on the line} (green) and the leader moves next to the next robot to be moved, (c) the leader moves the next robot to its position on the line while the following robot has its color set to \emph{following the leader} (brown), after which the leader moves to guide the next robot, and (d) the result of the Line Formation phase.}
  \label{fig:3.3}
\end{figure}

Now that one robot is in place, the leader can observe the location of this robot to ``remember'' where the line should be built. The leader proceeds to move the remaining robots to the line one at a time. To move a robot, the leader moves one unit to the left of the leftmost bottommost robot $r_i$ remaining on the convex hull and changes its color to \emph{follow the leader}. Robot $r_i$ changes its color to \emph{following the leader}, but does not move yet. The leader then repeatedly moves while robot $r_i$ follows the leader by moving to the leader's previous observed position until it reaches its position on the line. The movements of each robot are: (1) $r_i$ moves one unit to the left while the leader moves down until it reaches the $y$-coordinate of robot $r_i$'s intended position on the line, (2) $r_i$ moves down to where the leader stopped, while the leader moves left to $r_i$'s intended position on the line, and (3) the leader moves two units down and changes its color to \emph{do not follow} while $r_i$ moves to its position on the line and changes its color to \emph{on the line}. The leader now moves right to observe the remaining robots to find the new leftmost bottommost remaining robot. This process is repeated until every robot is positioned on the line. To avoid collisions, the leader ensures that there are two units of vertical space between consecutive robots on the line. We note that a robot can determine whether the phase has ended by seeing whether it can see any colors other than \emph{on the line} and any of the leader's colors. 

\subsection{Pattern Formation}
The goal of the Pattern Formation phase is to relocate the robots to form the given target pattern. The leader can determine where the robots should be placed by placing the first robot at $(0,0)$ with respect to its own coordinate system and building the pattern from there. The leader will build the pattern from this first robot towards the line, so given our assumed positioning this will be done from right to left, top to bottom. As the placement of the robots can be uniquely determined based on the robots on the line and the origin of the leader's coordinate system, the leader can determine what the last placed robot is and thus which part of the pattern has already been completed. During this phase, unless instructed otherwise by the leader, the robots stay on the line and do not move.

We present two versions of the Pattern Formation phase: one that allows scaling of the pattern (requiring one new color) and one that does not (requiring two new colors). 

\subsubsection{The Algorithm with Scaling of the Target Pattern}

In this version, we assume that the target pattern can be scaled. The first position of the first robot is at $(0,0)$ in the leader's coordinate system. Again, we explain our algorithm in the setting where the leader builds the pattern in the lowerleft quadrant of the origin, but the algorithm can be mirrored to obtain the other quadrants. This phase uses one new color to allow robots to remember that they have reached their final position in the pattern. 

After the Line Formation phase, all the robots are on a line with their light set to \emph{on the line}, except for the leader. Similar to the previous phase, the idea is that the leader moves a single robot to the pattern by moving next to it and guiding the robot to the intended position in the pattern. The leader moves the robots from top to bottom as ordered along the line. The following process is illustrated in Figure~\ref{fig:3.5}. To activate a robot $r_i$ on the line, the leader moves next to it and sets its color to \emph{follow the leader}. The leader then moves such that its $y$-coordinate is equal to that of the intended position in the pattern and in the next round the leader moves to the intended position in the pattern. Robot $r_i$ has set its color to indicate that it is following the leader and repeatedly moves to the leader's last observed position. To avoid collisions in the last step, the leader moves two units down and sets its color to \emph{do not follow} to indicate that the robot reached its final position. At this point, $r_i$ sets its color to \emph{at final position}. The leader now proceeds to pick up the next robot and repeats this process until the pattern is formed. If needed, the leader moves to fill the final position itself, if the pattern required exactly $n$ robots\footnote{We note that if the pattern requires more than $n$ robots, it cannot be built and all robots can detect this situation in the Leader Election phase and terminate at that point.}. 

\begin{figure}[ht!]
 \centering
  \includegraphics{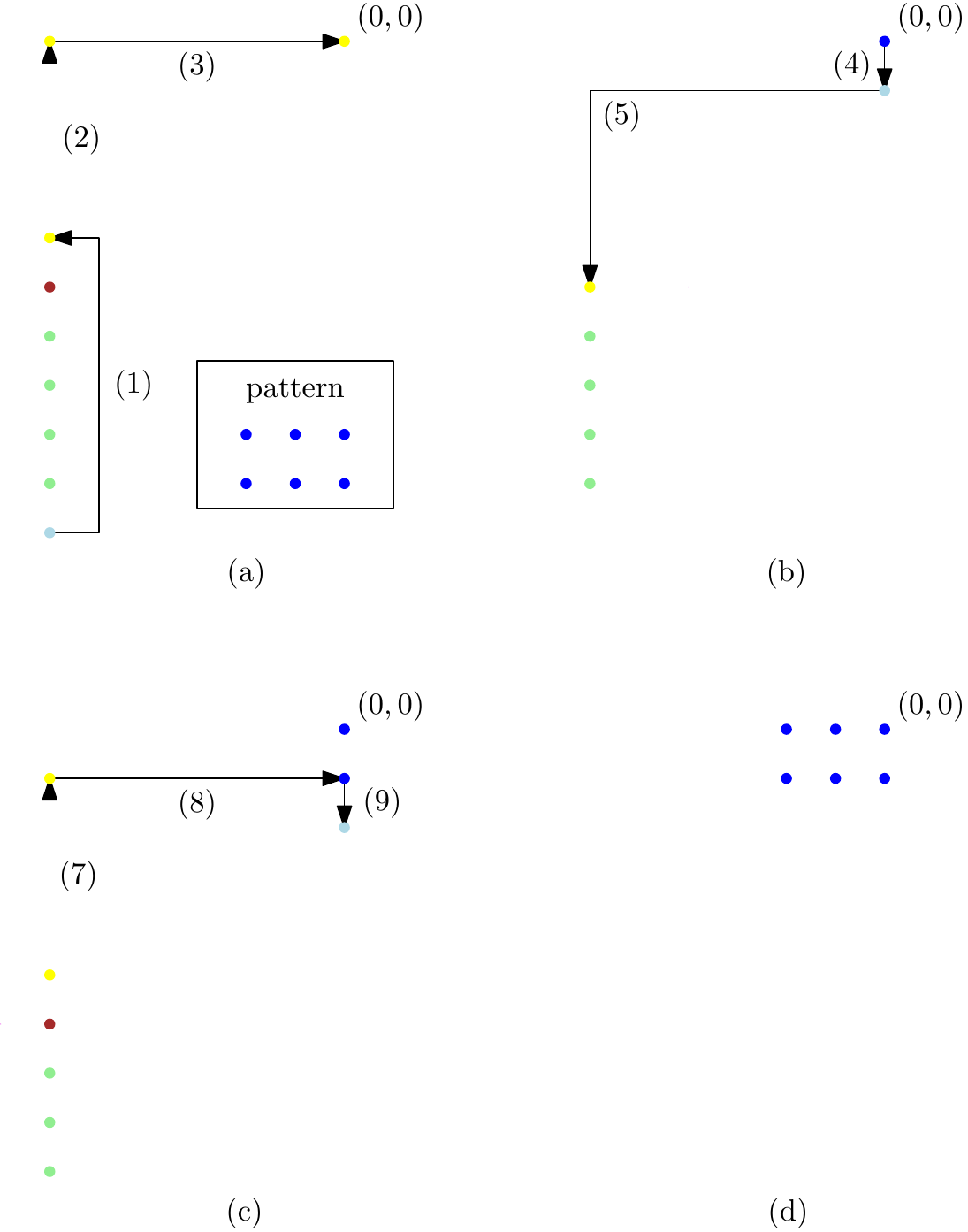}
  \caption{An example of the Pattern Formation phase when scaling is allowed (numbers indicate order of operations): (a) the leader moves to the position above the topmost robot $r_1$ on the line and signals that it should follow (yellow), which causes $r_1$ to change its color to \emph{following the leader} (brown) and move accordingly, (b) the leader moves out of the way so robot $r_1$ can reach its final position, which the leader indicates by changing its color to \emph{do not follow} (lightblue) and $r_1$ sets its color to \emph{at final position} (blue), after which the leader moves to be immediately above the next robot on the line, (c) the leader guides the next robot to its final position, and (d) all robots have reached their final position.}
  \label{fig:3.5}
\end{figure}

In order to ensure that after guiding a robot to its intended position the leader can indeed move down two units, we scale the pattern by a fixed constant factor, say $10$. Since in the given target pattern robots can at most be touching each other (otherwise the pattern cannot be constructed, which can be detected by all robots before the algorithm even starts), scaling the pattern this way ensures that there is ample space between the robots for the leader to move as described. 

\subsubsection{The Algorithm without Scaling of the Target Pattern}
In this version, we assume that the target pattern is not allowed to be scaled. As in the previous case, the pattern is built with respect to the origin $(0,0)$ in the leader's coordinate system and the pattern will be built in the lowerleft quadrant. This version of the phase needs two new colors: one to ``push'' robots, and one to allow robots to remember when they have reached their final position in the pattern.

The high-level idea in this version of the Pattern Formation phase is that the leader first ``pulls'' a robot behind it to position it in such a way that there is a straightline path from the robot to its intended position in the target pattern. Once the robot reaches this position, the leader moves away to indicate to the robot how far it needs to move in the direction \emph{opposite} to the direction the leader moved in, effectively ``pushing'' the robot to its final position. 

In more detail (see also Figure~\ref{fig:3.4}), in order to move a robot $r_i$, the leader starts by moving to the position that has the same $x$-coordinate as the line and same $y$-coordinate of the final position of $r_i$ in the target pattern. Here it changes its color to \emph{follow the leader}, signaling to the topmost robot $r_i$ of the line that it should move to the leader's current position. Robot $r_i$ observes this and changes its color to \emph{following the leader} to remember what it is supposed to do. Let $d$ be the distance between the leader's current position and the position in the pattern where it wants to place $r_i$. Note that since the leader is already at the $y$-coordinate of this intended location, this distance is just the difference in $x$-coordinate. After determining $d$, the leader moves $d$ to the left and changes its color to \emph{push}. This indicates to robot $r_i$ that in the next round it should move from its current position (where the leader used to be) to the position a distance of $d$ away from its current position in the direction opposite to where it sees the leader. Once robot $r_i$ reaches its final position, it changes its color to \emph{at final position} to remember this. Since every robot is moved using exactly one ``pull'' and one ``push'', they know that after the push they have arrived at their final position without the leader having to indicate this in any way. 

\begin{figure}[ht!]
 \centering
  \includegraphics[scale=0.95]{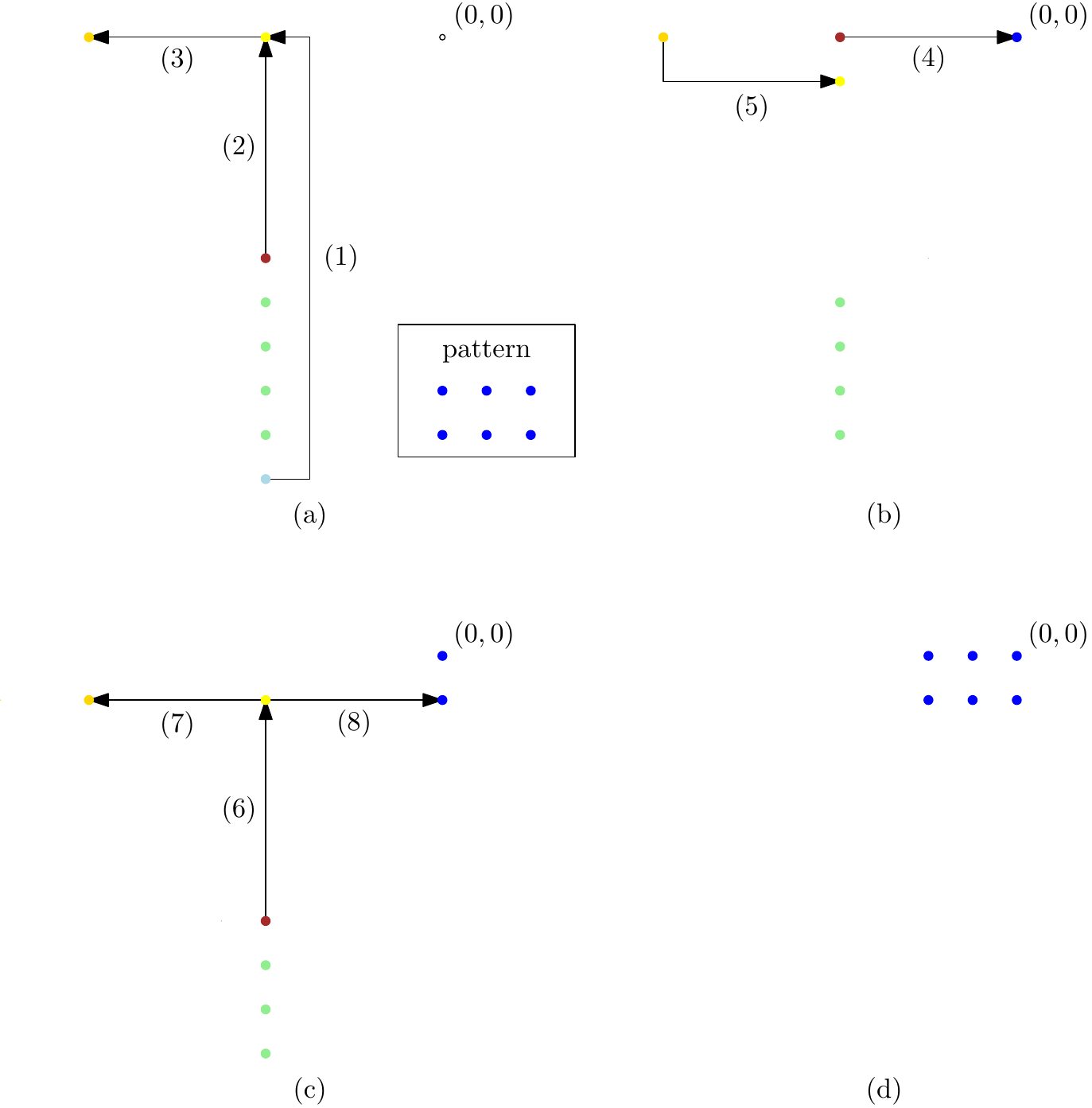}
  \caption{An example of the Pattern Formation phase when scaling is not allowed (numbers indicate order of operations): (a) the leader moves to the position on the line with $y$-coordinate equal to where the first robot $r_1$ needs to be placed and signals that $r_1$ should follow (yellow), which causes $r_i$ to change its color to \emph{following the leader} (brown) and move accordingly while the leader moves away preparing to push $r_1$, (b) the leader changes its color to \emph{push} (gold) and $r_1$ moves distance equal to its current distance to the leader away from the leader to reach its final position and sets its color to \emph{at final position} (blue), after which the leader moves to pull the next robot on the line, (c) the leader first pulls and then pushes the next robot to its final position, and (d) all robots have reached their final position.}
  \label{fig:3.4}
\end{figure}

The leader now moves to fetch the next robot and this process is repeated until the pattern is formed. As before, if needed, the leader moves to fill the final position itself.

\section{Analysis}
We proceed to prove that our algorithms solve the Pattern Formation problem for fat robots and that there are no collisions among the robots. Analogous to our algorithm description, throughout this analysis we will assume without loss of generality that the leader uses the lowerleft quadrant in the Line Formation and Pattern Formation phases.
We start by stating the result by Alsaedi {\it et al.}~\cite{https://doi.org/10.48550/arxiv.2206.14423} on the Mutual Visibility phase.

\begin{theorem}
\label{theorem:3.0}
Our algorithm solves the Mutual Visibility problem for unit disk robots in $O(n)$ rounds without collisions in the fully synchronous setting using two colors.
\end{theorem}

Next, we consider the Leader Election phase. This phase builds on the algorithm described and analyzed by Vaidyanathan {\it et al.}~\cite{vaidyanathan2022fast}, who bounded the expected number of rounds needed for this phase as well as the associated probability. The number of colors follows directly from needing three new colors to discern competing, non-competing, and leader robots.

\begin{theorem}
\label{theorem:3.1}
For any $q > 0$, Leader Election can be solved in the fully synchronous setting for $n$ robots in $O(q \log n)$ rounds with probability at least $1 - n^{-q}$ using three colors.
\end{theorem}

Next, we analyze the Line Formation phase.

\begin{lemma}
\label{lemma:3.22}
In every round of the Line Formation phase, only one robot that is not the leader moves from the convex hull to be positioned on the line, avoiding collisions.
\end{lemma}
\begin{proof}
Robots only move when the leader activates them to do so. If there are no robots on the line, all robots can see each other and thus they can determine whether they are the robot closest to the leader, resulting in only a single robot being activated. Once there are robots on the line, the leader moves next to a robot and sets its color to \emph{follow the leader} to activate it. As the leader activates the robots one at a time and completes moving a robot to the line before activating the next robot, only one robot moves from the convex hull to be positioned on the line. The ``no collisions'' part of the lemma then follows from the fact that we have only one moving robot, and that the robots are picked from left to right and start with a horizontal movement to move them away from the convex hull before moving vertically, thus ensuring there are no collisions.
\end{proof}

\begin{lemma}
\label{lemma:3.2}
The Line Formation phase uses four new colors.
\end{lemma}
\begin{proof}
The algorithm uses one color for the leader to indicate that a robot should follow it, one for a robot to store that it is following the leader, one for the leader to signal that the robot should stop following it, and one final color for the robot to store it is done for this phase.
\end{proof}

\begin{lemma}
\label{lemma:3.3}
The Line Formation phase takes $O(n)$ rounds.
\end{lemma}
\begin{proof}
In the Line Formation phase, the leader moves each robot from the convex hull to a line. The leader requires at most three rounds to move a robot to its position on the line. After each robot is moved, the leader uses at most three rounds to move next to the next robot. Therefore, the Line Formation phase takes $O(n)$ rounds to move all the robots from the convex hull to the line.
\end{proof}

Lemmas \ref{lemma:3.22}, \ref{lemma:3.2}, and \ref{lemma:3.3} imply the following theorem.

\begin{theorem}
\label{theorem:3.4}
The Line Formation phase takes $O(n)$ rounds, avoids collisions, and uses four new colors.
\end{theorem}

Next, we analyze the Pattern Formation phase. We start with the version where we can scale the pattern.

\begin{lemma}
\label{lemma:3.5}
In every round of the Pattern Formation phase with scaling, only one robot that is not the leader moves from the line to be positioned in the target pattern, avoiding collisions.
\end{lemma}
\begin{proof}
By construction, a robot only moves after it is activated by the leader. Since the leader moves the robots one at a time, only one robot will move. By moving the topmost robot of the line and starting by moving this robot vertically up (away from the other robots on the line), there are no collisions with other robots on the line. Furthermore, scaling the pattern by a large enough factor (such as $10$), we ensure that there are also no collisions with robots in the pattern either: the pattern is built from right to left, meaning that we can only collide with robots that were originally at most 
a distance of $1$ from the current robot. However, because the pattern is scaled, this distance is now increased to $10$, meaning that any overlap between the robot's paths into the pattern and any previously placed robots is removed.
\end{proof}

\begin{lemma}
\label{lemma:3.55}
The Pattern Formation phase with scaling uses one new color.
\end{lemma}
\begin{proof}
By construction, the algorithm uses only a single new color during this phase: for a robot to store that it has reached its final position.
\end{proof}

\begin{lemma}
\label{lemma:3.6}
The Pattern Formation phase with scaling takes $O(n)$ rounds.
\end{lemma}
\begin{proof}
In the Pattern Formation phase, the robots move from the line to the target pattern one by one. The leader moves at most three times to move next to the robot it wants to move and it uses three rounds to move this robot to its final position in the pattern. Hence, it takes a constant number of rounds to move one robot from the line to the target pattern. As a result, in total we require at most $O(n)$ rounds to move all the robots to the target pattern.
\end{proof}

Using Lemmas \ref{lemma:3.5}, \ref{lemma:3.55}, and \ref{lemma:3.6}, we get the following theorem.

\begin{theorem}
\label{theorem:3.7}
The Pattern Formation phase with scaling takes $O(n)$ rounds, avoids collisions, and uses one new color.
\end{theorem}

Using Theorems \ref{theorem:3.0}, \ref{theorem:3.1}, \ref{theorem:3.4}, and \ref{theorem:3.7}, we can now conclude the following.

\begin{theorem}
\label{theorem:3.8}
Our algorithm solves the Pattern Formation problem when scaling the target pattern is allowed for $n$ unit disk robots in $O(n) + O(q \log n)$ rounds with probability at least $1 - n^{-q}$ without collisions in the fully synchronous setting using 10 colors.
\end{theorem}

Finally, we analyze the Pattern Formation phase if scaling is not allowed.

\begin{lemma}
\label{lemma:3.5b}
In every round of the Pattern Formation phase without scaling, only one robot that is not the leader moves from the line to be positioned in the target pattern avoiding collisions.
\end{lemma}
\begin{proof}
By construction, a robot only moves once it is activated by the leader. Hence, only one robot moves at a time, and since the leader always picks the topmost robot on the line and moves it vertically away from the line, no collisions with the other robots on the line can occur. 
Since the leader knows the pattern and can determine the last robot placed based on the placement order of the pattern and the visible robots, it can also compute how to push the robot into the pattern to avoid collisions (since the pattern is built right to left from top to bottom, a left to right push exists). By pushing the robot this way, there are thus no collisions.
\end{proof}

\newpage
\begin{lemma}
\label{lemma:3.55b}
The Pattern Formation phase without scaling uses two new colors.
\end{lemma}
\begin{proof}
By construction, the algorithm uses two new colors: one to signal that a robot is being pushed, and one for the robot to store it has reached its final position.
\end{proof}

\begin{lemma}
\label{lemma:3.6b}
The Pattern Formation phase without scaling takes $O(n)$ rounds.
\end{lemma}
\begin{proof}
In the Pattern Formation phase, the robots move from the line to the target pattern one by one. The leader uses at most three moves to position itself above the topmost robot on the line and then another three rounds to move this robot to its final position in the pattern. Therefore, it takes a constant number of rounds to move one robot from the line to the target pattern. As a result, the total number of rounds used to move all the robots to the target pattern is $O(n)$.
\end{proof}

Using Lemmas \ref{lemma:3.5b}, \ref{lemma:3.55b}, and \ref{lemma:3.6b}, we get the following theorem.

\begin{theorem}
\label{theorem:3.7b}
The Pattern Formation phase without scaling takes $O(n)$ rounds, avoids collisions, and uses two new colors.
\end{theorem}

Using Theorems \ref{theorem:3.0}, \ref{theorem:3.1}, \ref{theorem:3.4}, and \ref{theorem:3.7b}, we obtain our final result.

\begin{theorem}
\label{theorem:3.8b}
Our algorithm solves the Pattern Formation problem when scaling the target pattern is not allowed for $n$ unit disk robots in $O(n) + O(q \log n)$ rounds with probability at least $1 - n^{-q}$ without collisions in the fully synchronous setting using 11 colors.
\end{theorem}

\section{Improving the Number of Colors}
\label{sec:improving_colors}
In this section, we improve the number of colors used to solve the Pattern Formation problem by reusing some of the colors used in the different phases discussed in the previous sections.

The Mutual Visibility phase uses two colors: \emph{off} for non-corner robots and \emph{corner} for corner robots. The Leader Election phase uses three new colors to keep track of status of the different robots as \emph{competing}, \emph{non-competing}, and \emph{leader}. We start by arguing that instead of using a new color for non-competing robots, we can use the color \emph{off} instead. 

\begin{lemma}
\label{lemma:3.9}
The \emph{off} color can be reused as the \emph{non-competing} color in the Leader Election phase.
\end{lemma}
\begin{proof}
We need to argue that the Leader Election phase still works as intended and that reusing the color does not cause any problems in later phases. 

When the robots activate in the Leader Election phase, they change their color to \emph{competing}. During this phase, unsuccessful robots change their color to \emph{off} as non-competing robots. During this process all robots are mutually visible and thus the non-competing robots can always see either a competing robot or the leader robot\footnote{If at any point all robots become non-competing, they could incorrectly conclude that they are in the Mutual Visibility phase, but since they would restart the Leader Election phase anyway, this is no issue.}, which have colors that help them identify the phase. Hence, they can conclude that they are in the Leader Election phase and thus act accordingly. 

Since some robots have the \emph{non-competing} color during part of the Line Formation phase, we now argue that changing this to the \emph{off} color does not cause any issues. We note that any robot that can see a robot with the \emph{leader} color or a robot with the \emph{on the line} color can conclude that it should not do anything unless the leader activates it and thus these robots cannot cause issues in the execution of the algorithm. However, if a robot sees neither of these colors it can mistakenly think that it is in the Mutual Visibility phase. We note that the robot would conclude that it is a corner and thus set its color to \emph{corner}. As a corner robot, it would move to expand the convex hull away from where the line is being built (as any robot that would expand the convex hull towards the line can see the line). As robots are removed from the convex hull by the leader, every robot will eventually see the line and at that point it can conclude the correct phase again and stop moving. We note that since the robots move to expand the convex hull away from the other robots, no collisions can occur. Hence, while the robots can mistake the phase they are in, the fact that they would stop the moment they see the leader or a robot on the line implies that this does not cause issues for our approach. 
\end{proof}

The Line Formation phase uses four new colors in order to allow the leader to activate a robot to follow it, for robots to store whether they are following the leader, for the leader to indicate that a robot should stop following it, and for a robot to store that it reached its position on the line. We argue that we can reuse the \emph{competing} color from the Leader Election phase as the color to indicate that a robot is following the leader.

\begin{lemma}
\label{lemma:3.10}
The \emph{competing} color can be reused as the \emph{following the leader} color in the Line Formation and Pattern Formation phase.
\end{lemma}
\begin{proof}
The \emph{competing} color was originally used to elect a single leader, so there was no leader before. Now since the robot with the \emph{competing} color sees the leader, it knows a leader has already been elected, and thus it can conclude that this is the Line Formation or Pattern Formation phase, where it has to follow the leader. Therefore, the \emph{competing} color can be reused as the \emph{following the leader} color in these phases.
\end{proof}

Next, we argue that we can also reuse the \emph{leader} color as the \emph{do not follow} color in these phases. 

\begin{lemma}
\label{lemma:3.11}
The \emph{leader} color can be reused as the \emph{do not follow} color in the Line Formation and Pattern Formation phase.
\end{lemma}
\begin{proof}
The function of both the \emph{leader} color in the Leader Election phase and the \emph{do not follow} color in the Line Formation and Pattern Formation phases is to indicate which robot is the leader, while ensuring that the other robots do not move. Hence, both colors allow the leader to move around without affecting moving the other robots, allowing it to get into position to guide the robots one at a time to their new positions on the line or in the pattern. 
\end{proof}

The above reuse of colors implies the following theorems. 

\begin{theorem}
\label{theorem:3.12}
Our algorithm solves the Pattern Formation problem when scaling the target pattern is allowed for $n$ unit disk robots in $O(n) + O(q \log n)$ rounds with probability at least $1 - n^{-q}$ without collisions in the fully synchronous setting using 7 colors.
\end{theorem}

\begin{theorem}
\label{theorem:3.13}
Our algorithm solves the Pattern Formation problem when scaling the target pattern is not allowed for $n$ unit disk robots in $O(n) + O(q \log n)$ rounds with probability at least $1 - n^{-q}$ without collisions in the fully synchronous setting using 8 colors.
\end{theorem}

\section{Conclusion}
We studied the Pattern Formation problem for unit disk robots in the robots with lights model under obstructed visibility. We described two algorithms for this problem, depending on the assumptions made to solve this problem. If the target pattern is allowed to be scaled with respect to the leader robot's coordinate system, our initial algorithm used 10 colors, which we subsequently improved to 7 colors. If scaling is not allowed, our algorithm needs one additional color: initially 11, which we then improved to 8. 

Our algorithms run in $O(n) + O(q \log n)$ rounds with probability at least $1 - n^{-q}$ in the fully synchronous model. Interestingly, unlike previous work, our algorithms do not require any additional assumptions on the capabilities of the robots or any shared information or coordinate system. 

There are a number of interesting directions in which we could consider extending this work. For example, we have no lower bounds indicating that the number of colors we use is optimal and thus the natural open problems are both trying to improve the number of colors or showing that this is not possible using a lower bound. 

Another interesting direction is moving to a semi-synchronous or even asynchronous model. Given the strict ordering required in our algorithms, this is challenging, as our ``push'' and ``pull'' operations use the leader's position to determine where to move. Some of these issues may be resolved by the leader not moving until the robot following it is at the correct position, but ensuring that this does not increase the number of rounds drastically or causes collisions is not straightforward. 

Returning to the classical model (without lights) it would also be interesting to determine whether the Pattern Formation problem can be solved efficiently in this model as well or whether additional assumptions are required.


\bibliography{references}

\end{document}